\documentclass[prl, twocolumn, aps,floats,preprintnumbers,amssymb,nofootinbib]{revtex4-1}
\usepackage{microtype,graphicx}
\usepackage[english]{babel}
\usepackage{amsthm}
\usepackage{enumitem}
\usepackage{latexsym}
\usepackage{amsfonts}
\usepackage{amssymb}
\usepackage{caption}
\usepackage[dvipsnames]{xcolor}
\usepackage{CJK}
\usepackage[export]{adjustbox}
\usepackage{amsmath}
\usepackage{slashed}
\usepackage{dcolumn}
\usepackage{verbatim}
\usepackage{appendix}
\usepackage{float}
\usepackage{multirow}
\usepackage{soul}
\usepackage[normalem]{ulem}
\usepackage{ulem}
\usepackage{hyperref}
\usepackage{adjustbox}
\usepackage{natbib}
\usepackage{epsfig}
\usepackage{nicefrac}

\newcommand{\noun}[1]{\textsc{#1}}

\newtheorem{thm}{\protect\theoremname}
  \theoremstyle{plain}
  \newtheorem{prop}[thm]{\protect\propositionname}

  \newtheorem{conj}[thm]{\protect\conjecturename}
   \newtheorem{lem}[thm]{\protect\lemmaname}
   \newtheorem{rmk}[thm]{\protect\remarkname}

\usepackage[english]{babel}
  \providecommand{\definitionname}{Definition}
  \providecommand{\propositionname}{Proposition}
  \providecommand{\theoremname}{Theorem}
  \providecommand{\corollaryname}{Corollary}
  \providecommand{\conjecturename}{Conjecture}
   \providecommand{\lemmaname}{Lemma}
   \providecommand{\remarkname}{Remark}

\DeclareMathOperator{\arcsinh}{arcsinh}
\DeclareMathOperator{\arctanh}{arctanh}

\def\beq{\begin{equation}}
\def\eeq{\end{equation}}
\def\be{\begin{equation}}
\def\ee{\end{equation}}
\def\bea{\begin{eqnarray}}
\def\eea{\end{eqnarray}}

\newcommand{\eq}[1]{Eq.~(\ref{#1})}

\newcommand{\fig}[1]{Fig.~\ref{#1}}

\newcommand{\clight}{c_{\mathrm{light}}}

\begin{document}

\title{Inflationary resolution of the initial singularity}
\author{Damien A. Easson\footnote{easson@asu.edu}}
\affiliation{
Department of Physics \& Beyond Center for Fundamental Concepts in Science,  
Arizona State University, Tempe, AZ 85287-1504, USA}
\author{Joseph E. Lesnefsky\footnote{jlesnefs@asu.edu}}
\affiliation{ 
Beyond Center for Fundamental Concepts in Science,  
Arizona State University, Tempe, AZ 85287-1504, USA}

\date{\today}

\begin{abstract}
The inflationary paradigm has transformed our understanding of the early universe; yet most inflationary models are considered geodesically past-incomplete, suggesting a beginning of time or a primordial Big Bang singularity. The Borde–Guth–Vilenkin (BGV) theorem is often cited as demonstrating that all eternally inflating spacetimes must be past-incomplete. Utilizing a new theorem establishing geodesic completeness in generalized cosmologies, we present a simple, explicit class of inflationary solutions that are smooth, nonsingular, and geodesically complete for all time, including into the past. These models exhibit localized NEC violation but remain globally well-behaved in both temporal directions. The NEC violation is confined and controlled: the averaged null energy condition (ANEC) is satisfied in the strongest sense, while violations of smeared null energy conditions (SNEC) are uniformly bounded and become nonnegative under sufficiently wide smearings. Our results suggest that eternal inflation can arise from controlled NEC-violating dynamics, offering a new, nonsingular, and past-eternal picture of the universe.
\end{abstract}
\maketitle
\emph{Introduction}--The inflationary universe paradigm is a cornerstone of modern cosmology~\cite{Guth:1980zm,Linde:1981mu,Albrecht:1982wi}. A prevailing view asserts that inflationary scenarios cannot be past-eternal even at the classical level, a conclusion drawn independent of the energy conditions involved. Specifically, the general notion stems largely from the renowned work of Borde, Guth and Vilenkin (BGV) \cite{Borde:2001nh}, who stated that a cosmological model which is inflating--or just expanding sufficiently fast--must be incomplete in null and timelike past directions. This belief has led to the strong assertion that inflationary models require new physics beyond inflation itself to describe the past boundary of the inflating region, and even to the broader view that, not only did inflation have a beginning, but  the universe itself must have originated from a definite beginning.

In this Letter, we challenge this perspective by explicitly constructing geodesically-complete eternal inflationary models. More broadly, we show that all non-trivial geodesically-complete Friedmann Robertson Walker (FRW) spacetimes \it necessarily \rm require an epoch of accelerated expansion. This finding underscores the critical role of inflation-like dynamics in ensuring geodesic completeness.~\footnote{Here and throughout, ``non-trivial FRW'' refers to cosmologies with a genuinely evolving (nonstatic) scale factor. 
Static spacetimes, such as Minkowski or the Einstein static universe, are therefore excluded by assumption and do not constitute counterexamples to the propositions discussed in this Letter.}

The price of this eternal inflation in General Relativity (GR) is a period of null energy condition (NEC) violation. However, in positively curved FRW universes we construct nonsingular, geodesically complete, eternally inflating spacetimes in which this violation is tightly controlled: the averaged null energy condition (ANEC) is satisfied in the strongest sense, while the smeared null energy condition (SNEC) is obeyed as a quantum inequality, with any negative smeared averages bounded and vanishing under sufficiently wide smearings.
Further details concerning inflationary, as well as geodesically--complete bouncing and loitering models, which may require only an arbitrarily short period of accelerated expansion, are discussed in our companion work \cite{Easson:2024fzn}, (see also \cite{Burwig:2025hrr}).

\emph{Eternal inflating universe}--We begin with a detailed analysis of the model introduced in 
\cite{Lesnefsky:2022fen}, having scale factor:
\begin{eqnarray}\label{aplusc}
    a(t) = a_0 \exp[2 t/\alpha] + c \,,
\end{eqnarray}
for constants $a_0$, $\alpha$ and $c$. We refer to this model as the ``plus c" model. For this scale factor, $2 \alpha^{-1} \neq H = \dot{a}/a$; and
$c>0$ is required to construct a geodesically complete spacetime.\footnote{An earlier model which shares some feature of the above was discussed in \cite{Ellis:2002we}. As we shall discover, the constant $c>0$ plays the role of a
nonsingularity regulator. It fixes the minimal radius of the universe,
$a_{\min}=c$, so that the spatial sections never collapse to zero size. In
the $k=1$ case the past limit is the Einstein--static universe
$\mathbb{R}\times S^3(c)$, with curvature scale set by $c$. Physically, $c$
encodes the minimal curvature radius of the nonsingular past state, and
ensures that the ANEC is satisfied in the strongest
sense.
} 

We use natural units with $\hbar=\clight=1$. The symbol $c$ in \eq{aplusc} is a constant controlling the minimum scale factor, and $\alpha$ is an arbitrary-scale parameter with dimensions of Mass$^{-1}$. We express all quantities in reduced Planck units by setting 
the Planck mass $M_{pl}= 1/\sqrt{8 \pi G}$=1.
In what follows, including plot parameters, we take $a_0 = c = 1/2$, and FRW spatial curvature $k=1$. The special role of positive spatial curvature in enabling geodesically complete cosmologies consistent with ANEC is analyzed in detail in~\cite{Burwig:2025hrr}. With these values it is  easy to show \eq{aplusc} may equivalently be expressed as
\begin{equation}
    a(t) = \exp[t/\alpha] \cosh[t/\alpha] \,.
\end{equation}

The Hubble parameter $H = \dot a/a$ is given by:

\begin{equation}
    H = \frac{2 a_0 e^{\frac{2t}{\alpha}}}{(c + a_0 e^{\frac{2t}{\alpha}}) \alpha} \,.
\end{equation}
This is an example of an eternally inflating spacetime, as is easily seen from the eternal positive acceleration 

\begin{equation}
  \frac{\ddot a}{a} = \frac{4}{\alpha^2} \left(1 - \frac{c}{c + a_0 e^{\frac{2t}{\alpha}}}\right) \,.
\end{equation}
\begin{figure}[H]
\centering
  \includegraphics[width=1\linewidth]{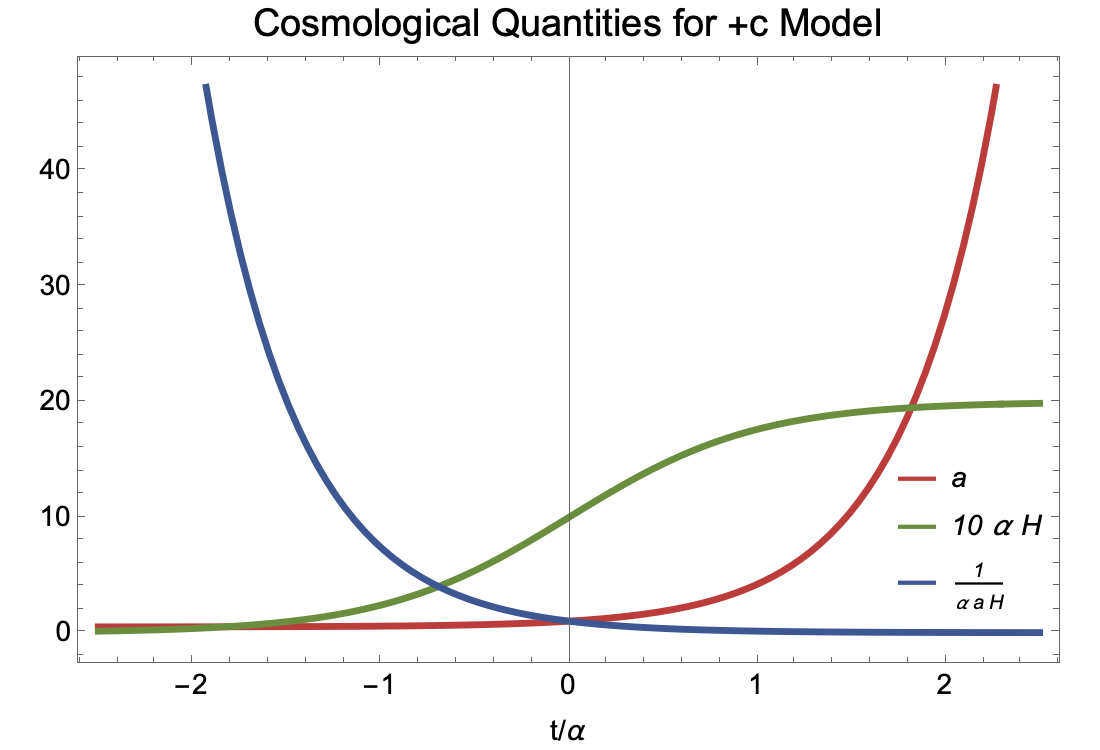}
  \captionof{figure}{Evolutionary behavior of cosmological parameters.
The scale factor $a\left( t \right)$ is plotted in red. The dimensionless Hubble parameter $\alpha H$ multiplied by 10 is in green. The dimensionless co-moving Hubble radius, decreasing for all time, is in blue. }
  \label{ceplusc}
\end{figure}
\begin{figure}[H]
\centering
  \includegraphics[width=1\linewidth]{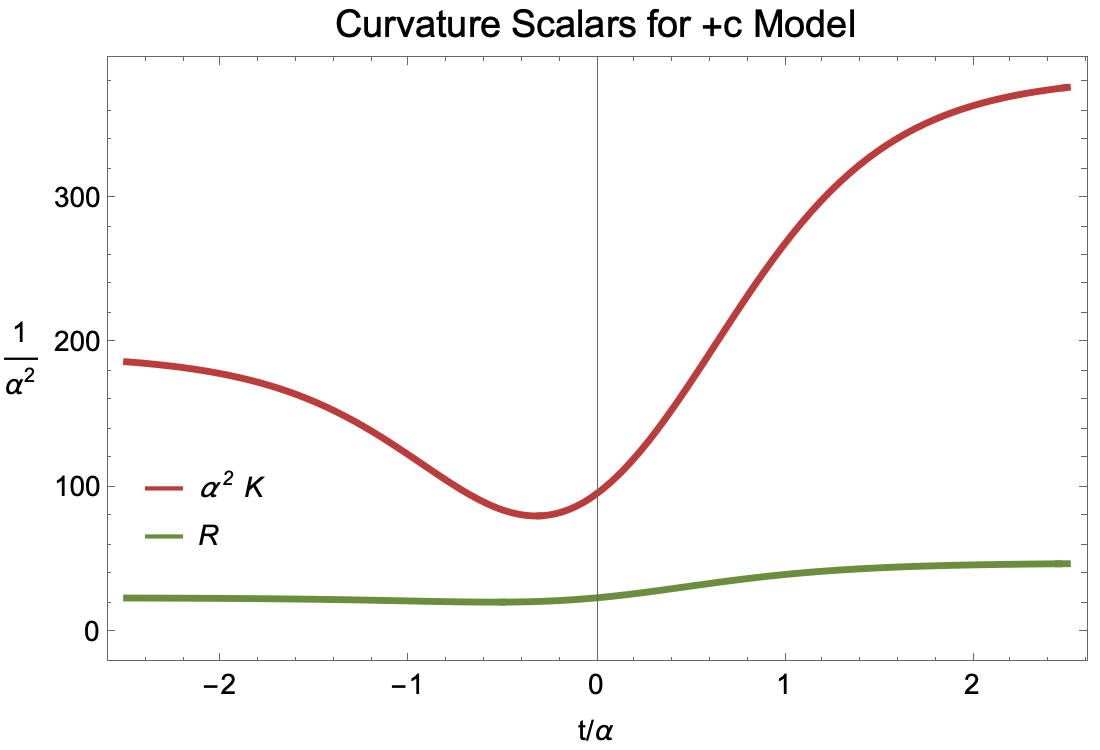}
  \captionof{figure}{Curvature scalars.
  The Kretschmann scalar is plotted in red. The Ricci scalar is in green.  All curvatures have been renormalized to have units of $\nicefrac{1}{\alpha^2}$, and remain finite for all times.}
  \label{curvplusc}
\end{figure}
The model is eternally accelerating for all $c>0$ and nonzero $\alpha$. 
For $\alpha>0$, the universe expands monotonically and inflates at all times. 
In the infinite past ($t \to -\infty$) the scale factor approaches a constant 
$a \to c$ and the spacetime asymptotes to the Einstein--static universe with 
finite curvature $R=6/c^{2}$. In the infinite future ($t \to +\infty$) the 
expansion becomes asymptotically de Sitter, with
$H \to 2/\alpha$ and $R \to 24/\alpha^{2}$. Thus the geometry 
interpolates smoothly between a nonsingular, static past state and a future 
de Sitter phase.
For $\alpha<0$, the model is eternally contracting, and inflating, approaching zero acceleration as $t \rightarrow \infty$.

The cosmological evolution is depicted in \fig{ceplusc}. Shown are the scale factor $a(t)$, Hubble parameter $H$ and co-moving Hubble radius $H^{-1}/a$. When the co-moving Hubble radius is decreasing the spacetime is inflating--in this case it is decreasing for all time.

The absence of curvature singularities is confirmed by the finiteness of curvature invariants, as illustrated in Fig.~\ref{curvplusc}. We plot the Ricci scalar $R$ and the Kretschmann scalar $K = R_{\mu\nu\rho\sigma}R^{\mu\nu\rho\sigma}$, built from the Riemann tensor:


\begin{eqnarray}
R &=& \frac{6 + \tfrac{24X(c+2X)}{\alpha^{2}}}{(c+X)^{2}} \,, \nonumber\\[4pt]
K &=& \frac{12\!\left[\tfrac{16X^{2}(c+X)^{2}}{\alpha^{4}}+\bigl(\tfrac{4X^{2}}{\alpha^{2}}+1\bigr)^{2}\right]}{(c+X)^{4}} \,,
\end{eqnarray}
where, $\quad X\equiv a_0 e^{2t/\alpha},\; a(t)=c+X \,.$
As shown, both quantities remain finite for all cosmic times $t$.

\emph{Geodesic Completeness}--Our discussion includes Generalized Friedmann-Robertson-Walker (GFRW) spacetimes, of which FRW models used in modern cosmology are an exceptional subset.
A time dimension of $\mathbb{R}$ is warped with smooth strictly positive scale factor $a>0$ to any geodesically complete Riemannian manifold constituting the purely spacelike foliation.  The FRW spacetime has a spatial section of constant sectional curvature $k$.  The geodesic completeness of GFRWs is completely determined by the behavior of $a$, as discussed in \cite{Lesnefsky:2022fen}:
 \begin{thm}{\noun{(Lesnefsky, Easson, Davies - (LED))}}--Consider a GFRW spacetime.
        \begin{enumerate}
            \item The spacetime is future timelike complete if and only if $ \int_{t_0}^\infty \frac{a \left( t \right) dt}{\sqrt{\left( a \left( t \right) \right)^2 + 1}}$ diverges for all $t_0 \in \mathbb{R}$. 
            \item The spacetime is future null complete if and only if $ \int_{t_0}^\infty a \left( t \right) dt$ diverges for all $t_0 \in \mathbb{R}$. 
            \item $\mathcal{M}$ is \emph{future spacelike complete} iff it is future null complete and
the warping function is bounded: $a<\infty$.
            \item The GFRW is past timelike / null / spacelike complete if, for items 1-3 above, upon reversing the limits of integration from $\int_{t_0}^\infty$ to $\int_{-\infty}^{t_0}$ the word ``future'' is replaced by ``past''.
            \item The spacetime is \emph{geodesically complete} if and only if it is both future and past timelike, null, and spacelike geodesically complete.
        \end{enumerate}
        \label{ledthm}
    \end{thm}
%

Unlike the BGV theorem \cite{Borde:2001nh} to be discussed below, which purports only to show geodesic incompleteness, Thm.~\ref{ledthm} represents a significant advancement, offering a concrete method for ascertaining the geodesic completeness, or incompleteness, of a specific FRW spacetime.

For the scale factor of Eq. \ref{aplusc}, it is possible to explicitly calculate the integrals of Thm.~\ref{ledthm}. Assuming $c>0$ and $a_0>0$, we find for the indefinite integrals:
\begin{align}
&\int^t \frac{a(\zeta)}{\sqrt{(a(\zeta))^2 + 1}} \, d\zeta =
\frac{1}{2} \alpha \arcsinh \left(c + a_0 e^{\frac{2t}{\alpha}}\right) \nonumber \\
&+ \frac{ \alpha \,c}{\sqrt{1 + c^2}} \arctanh \left(\frac{a_0 e^{\frac{2t}{\alpha}} - \sqrt{1 + \left(c + a_0 e^{\frac{2t}{\alpha}}\right)^2}}{\sqrt{1 + c^2}}\right)
\end{align}
and
\begin{equation}
\int^t a \left( \zeta \right) d\zeta= ct + a_0 \frac{\alpha}{2} e^{\frac{2t}{\alpha}} \,. 
\end{equation}

It is easy to show the above integrals diverge over the full set of conditions discussed in Thm.~\ref{ledthm} for all (non-zero) values of $\alpha$; hence, the spacetime with scale factor \eq{aplusc} is geodesically complete. 
  
We have thus demonstrated that the FRW spacetime defined by the scale factor in Eq.~\ref{aplusc} is geodesically complete, eternally inflating, and nonsingular---thereby providing a concrete counterexample to the prevailing interpretation of the BGV theorem. This result shows that past-complete inflationary models can be constructed within classical general relativity. 
Such models are not isolated curiosities: additional examples of geodesically complete cosmologies are presented in~\cite{Easson:2024fzn}. 
For example, in the limiting-curvature constructions of~\cite{Chamseddine:2016uef}, the FRW scale
factor remains finite for all cosmic times and undergoes a smooth,
nonsingular evolution, providing an explicit realization of a
geodesically complete cosmology. The model exhibits both a bounce and a period of accelerated expansion in line with our Conjecture 5 below.~\footnote{
Although the scale factor in Ref.~\cite{Chamseddine:2016uef} is originally derived
within a modified-gravity framework, the same background evolution may
equivalently be embedded at the level of homogeneous FRW dynamics in standard
general relativity in a closed $(k=1)$ universe and sourced by a canonically
normalized scalar field with an appropriate potential, using standard
reconstruction techniques.
}

 Solutions of this type illustrate that the model studied in this Letter is representative of a broader
class of nonsingular FRW spacetimes. We refer the reader to \cite{Easson:2024fzn,Burwig:2025hrr}
for further examples and a systematic analysis.


\begin{figure}[H]
\centering
  \includegraphics[width=1\linewidth]{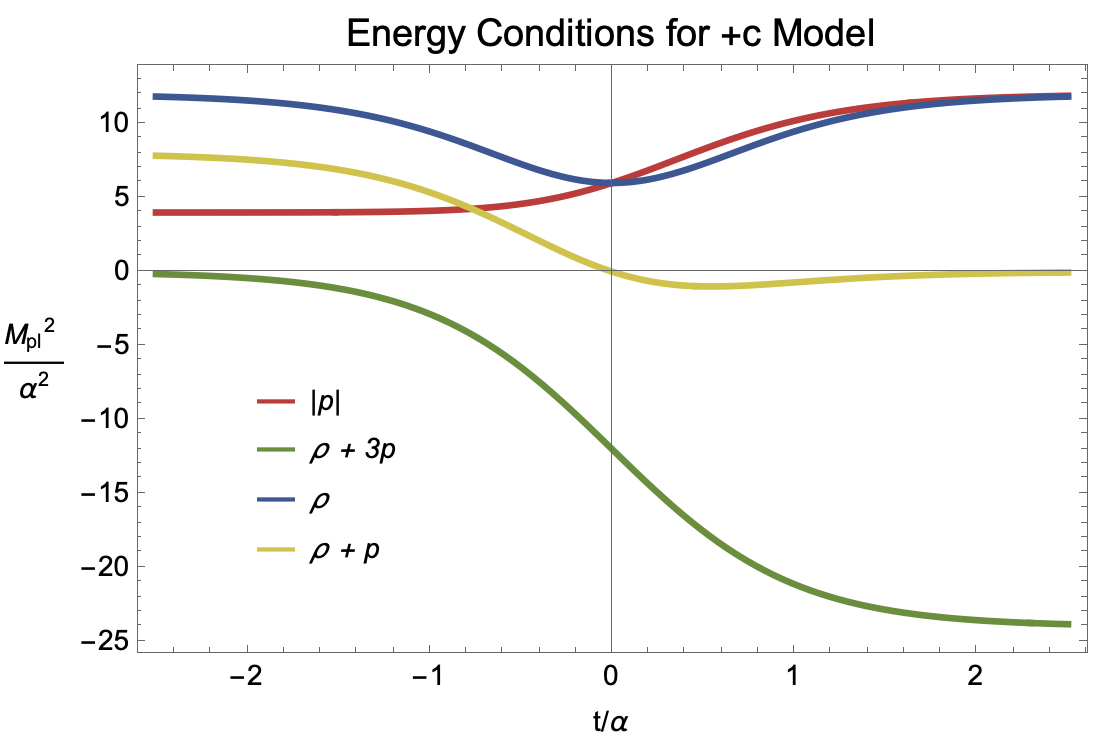}
 \captionof{figure}{
Energy conditions from \eq{etplusc}. Plot of the energy density $\rho$ (blue),
$\rho + p$ (yellow), $|p|$ (red), and $\rho + 3p$ (green).
At early times ($t\to -\infty$), $\rho+p$ approaches a constant positive value,
while at late times ($t\to +\infty$) it asymptotically decays to zero,
illustrating the saturation behavior of the NEC.
}
  \label{ecplusc}
\end{figure}
\emph{Energy conditions}--While both Thm.~\ref{ledthm} and the BGV theorem report to hold independent of the energy conditions, such an analysis is informative and we now examine the energy conditions for the model given by Eq. \ref{aplusc}. Calculation of the Einstein tensor yields non-vanishing components:
\begin{eqnarray}\label{etplusc}
  G_{tt} &=& \frac{3 + \tfrac{12 a_0^{2} e^{\tfrac{4t}{\alpha}}}{\alpha^{2}}}
                  {\left(c + a_0 e^{\tfrac{2t}{\alpha}}\right)^{2}}\,, \nonumber\\[4pt]
  G_{ii} &=& -\,\frac{1 + \tfrac{4 a_0 e^{\tfrac{2t}{\alpha}}\!\left(2c + 3 a_0 e^{\tfrac{2t}{\alpha}}\right)}{\alpha^{2}}}
                  {\left(c + a_0 e^{\tfrac{2t}{\alpha}}\right)^{2}}\,g_{ii}\,,
\end{eqnarray}
The energy density is given by $\rho = - G^t{}_t$ and the pressure is $p = G^i{}_i$. A plot elucidating the energy conditions is given in \fig{ecplusc}.

To interpret this plot we recall that in the standard perfect-fluid treatment of matter where $p$ and $\rho$ are the pressure and energy density of the fluid respectively, the \emph{energy conditions} in a cosmological setting are as follows \cite{S.W.Hawking1973}:
Weak energy condition (WEC): $\rho\ge 0$ and $\rho + p \ge 0$, Null energy condition (NEC): $\rho + p \ge 0$, Dominant energy condition (DEC): $\rho\ge |p|$, and Strong energy condition (SEC): $\rho + p \ge 0$ and $\rho + 3p \ge 0$. Note that the DEC$\implies$WEC, WEC$\implies$NEC, SEC$\implies$NEC; and SEC$\rlap{$\quad\not$}\implies$WEC.

Our stance on employing energy conditions as a stringent basis for critiquing models is marked by ambivalence. Notably, all classical energy conditions are unequivocally breached by quantum effects, a fact supported both experimentally and theoretically, as exemplified by the Casimir Effect \cite{Casimir:1948dh} and outlined in \cite{Epstein:1965zza}. Consequently, we entrust the assessment of solution viability to the discernment of our readers, refraining from making definitive judgments \cite{Barcelo:2002bv,Chatterjee:2012zh}.

From Fig.~\ref{ecplusc} we see that each of the classical energy conditions is
violated at some stage of the eternal inflationary model defined by
Eq.~\eqref{aplusc}. In particular, violation of the NEC is indicated by the
yellow curve dipping below the horizontal axis; this violation is confined to
late times. The NEC is satisfied at early times and approaches saturation
asymptotically. Analytically, one finds $\rho+p \to 2/c^{2}$ as
$t\to -\infty$ and $\rho+p \to 0$ as $t\to +\infty$, in agreement with the
behavior shown in Fig.~\ref{ecplusc}. Since the model is eternally inflating,
with $\ddot a>0$ for all time, the strong energy condition is violated
throughout the evolution, as reflected by the persistent negativity of the
green curve.

Hence, the price of realizing this eternal, nonsingular, inflating universe
within classical GR is a temporary, future violation of the null energy condition
(NEC). While such violations are often associated with instabilities in
simple perfect-fluid models, it is now well established that controlled and
stable NEC violation can arise in consistent effective field theories.
Explicit examples include ghost-condensate models, galileon and generalized
Horndeski theories, and related higher-derivative scalar EFTs, in which NEC
violation occurs without the introduction of ghosts or gradient
instabilities within their regime of validity
\cite{Dubovsky:2005xd,Nicolis:2009qm,Kobayashi:2010cm,Easson:2011zy,
Battefeld:2014uga,Easson:2018qgr}.

From this perspective, the stress tensor appearing in Eq.~(8) should be
understood as an effective description of a scalar sector with derivative
self-interactions, rather than as a fundamental fluid. We emphasize that the
purpose of the present construction is not to commit to a specific
microphysical Lagrangian, but to demonstrate that the geometry and stress
tensor required for geodesic completeness and eternal inflation are fully
compatible with known classes of stable NEC-violating EFTs.

In particular, backgrounds closely related to the plus-c scale factor can arise as homogeneous cosmological solutions of galileon-type and ghost-condensate theories, in which NEC violation is localized in time and does not destabilize the background evolution. Thus, while we do not attempt a full EFT reconstruction here, the plus-c model lies squarely within
the landscape of cosmological backgrounds that admit a consistent effective
field theory completion. In other models the duration of the NEC violating
interval can be made arbitrarily short, potentially of order the Planck
time, without spoiling geodesic completeness \cite{Easson:2024fzn}.

We now examine two integrated energy conditions that provide even stronger diagnostics: the averaged null energy condition (ANEC) \cite{PhysRevD.51.4277, Wald:1991xn} and the smeared null energy condition (SNEC)  \cite{Freivogel:2018gxj,Moghtaderi:2025cns}.
The ANEC serves as a powerful diagnostic for distinguishing between “benign” and “pathological” violations of the NEC. While the NEC may be locally violated in physically reasonable settings–such as in semi-classical or quantum field theory–these violations are often harmless when they occur in small, localized regions, provided they are compensated by positive energy elsewhere along the same null geodesic. The ANEC formalizes this by requiring that the integral of $T_{\mu\nu} k^\mu k^\nu$ along a complete null geodesic remain nonnegative. Unlike the pointwise NEC, the ANEC is known to hold in a wide range of well-behaved quantum field theories in flat spacetime and is often preserved even in curved backgrounds, so long as extreme phenomena such as traversable wormholes or closed timelike curves are absent. As such, the ANEC provides a more robust and physically meaningful constraint on energy densities than the NEC alone.

The ANEC is given by averaging along complete null geodesics with normalization \( k^t = \frac{dt}{d\lambda} = \frac{1}{a} \). In the case of \eq{aplusc}, the integral yields:
\beq
\int_{-\infty}^{\infty} T_{\mu\nu} k^\mu k^\nu\, d\lambda
= \int_{-\infty}^{\infty} \frac{\rho + p}{a(t)}\, dt = +\infty,
\eeq
so the ANEC is satisfied in the strongest sense. This implies averaged focusing via the Raychaudhuri equation and is consistent with the absence of standard exotic causal structures under customary global assumptions, despite late-time NEC violation.

We further evaluate a quantum-inequality-inspired smeared null energy condition (SNEC),
\beq
\int f(\lambda)\, T_{\mu\nu} k^\mu k^\nu\, d\lambda \;\geq\; -\frac{C}{\ell^4},
\eeq
where \( f(\lambda) \) is a smooth sampling function of affine width \( \ell \). Using Gaussian smearings for $f$ and the same null normalization, we find that the smeared average becomes negative at late times (reflecting the background NEC violation) but is rapidly suppressed; across a broad range of centers and widths the data are compatible with a dimensionless constant \( C = \mathcal{O}(1) \). We find that while the strict SNEC can be violated by localized smearings supported in the late-time NEC-violating region, the violation is quantitatively bounded and physically mild. Moreover, for widening smearings that sample a sufficient portion of the positive-energy past, the smeared null average becomes nonnegative. These results demonstrate controlled, non-pathological SNEC violation, fully consistent with the strong satisfaction of the ANEC (see Appendix).

The presented spacetime is geodesically complete, eternally inflating, and nonsingular--directly contradicting the widely held notion that inflationary cosmologies must be incomplete, regardless of energy condition violation. This model satisfies the ANEC strongly, while exhibiting only mild, localized violations of the standard NEC, consistent with the Smeared NEC (SNEC) under finite smearing. These diagnostics suggest that the NEC violation is physically controlled and non-pathological at the level of integrated energy conditions. Consequently, we have proven certain inflationary models are capable of evading the initial cosmological singularity without invoking quantum gravity or exotic boundary conditions.~\footnote{Full dynamical stability of perturbations depends on the microphysical completion, to be studied in future work.}

\emph{BGV Theorem}--The results presented above may appear to conflict with the widely cited no-go theorem of BGV, which is often interpreted as ruling out the possibility of past-eternal inflation. We therefore examine the BGV theorem in the context of the model defined by \eq{aplusc}.
The theorem, Eq.~(5) of \cite{Borde:2001nh}, may be quantified:
\begin{thm} \label{thm:bgv} 
\emph{\noun{(Borde, Guth, Vilenkin - (BGV))}}--Consider a spacetime.  Let
$\gamma$ be some causal geodesic defined over domain $\left[\lambda_{i},\lambda_{f}\right]$.~~If the quantity
\begin{equation}
H_{avg}^{\gamma}=\frac{1}{\lambda_{f}-\lambda_{i}}\int_{\lambda_{i}}^{\lambda_{f}}H^{\gamma}\left( \zeta \right)d\zeta\label{eq:havg bgv}
\end{equation}
is strictly positive along the image of $\gamma$, the spacetime is geodesically-incomplete.
\end{thm}
\noindent
As we discuss below, $H_{avg}^{\gamma}$ should be understood as an asymptotic quantity evaluated on maximal past-directed geodesics (see footnote 6).
 
Without loss of generality we take $a_0=1$, $t_i < 0$, and $t_f = 0$, and select any connected interval $\left[ t_i , 0 \right]$ where the boundary is actually realized:
\begin{equation}
     H_{avg}^{\gamma}=  \frac{1}{-t_i} \int_{t_i}^0 H dt = \frac{1}{-t_i} \ln \left( \frac{1 + c}{e^{\nicefrac{t_i}{\alpha}} + c}\right) >0 \label{eq:havg non lim inf}
\end{equation}
because, $1 + c > e^{\nicefrac{t_i}{\alpha}} +c$. Thus, direct calculation of \eq{eq:havg bgv}, yields $H_{avg}>0$, yet despite Thm.~\ref{thm:bgv}, the model is geodesically complete per Thm.~\ref{ledthm}. This apparent tension exposes a deeper issue: In Thm.~\ref{thm:bgv},  \eq{eq:havg bgv} is computed over compact intervals; on maximal past rays the averaged rate tends to
0, so the hypothesis fails. A proper discussion of geodesic completeness should involve maximal geodesics or maximal geodesic rays, as we have shown above. A geodesic defined over a compact interval is \it inherently \rm incomplete, as it must inevitably encounter a singularity or boundary at its endpoint, rendering it categorically incomplete yet inextendable, or it is straightforwardly extendable in a (possibly small but non-empty) open neighborhood of the endpoint by the exponential map \cite{Lesnefsky:2022fen}.

One may assume that the authors of Thm. \ref{thm:bgv} intended that the limit $t_i \rightarrow -\infty$ be taken; although, no such limits were explicitly discussed in \cite{Borde:2001nh}, leaving the original formulation ambiguous.~\footnote{A precise, nonambiguous hypothesis is, for example,
\[
\varliminf_{L\to\infty}\,\frac{1}{L}\int_{-L}^{0} H^\gamma(\lambda)\,d\lambda > 0\,,
\]
for a past-directed affine parameter $\lambda$ along $\gamma$. This allows fluctuations but demands positivity that persists as the averaging window extends to the infinite past.
Our model gives a vanishing liminf, so it lies outside the  scope of the theorem.}

This limiting case would correspond to a past-directed maximal geodesic ray. Naturally, one may further consider a geodesic maximally extended in both temporal directions. 
For such a past-directed maximal geodesic,  \eq{eq:havg non lim inf} yields $H_{avg}=0$, and one may then argue that the BGV theorem does not apply. This result is due to the cofinite interval suppression of $1/(-t_i)$ in the integral. Hence, with this limit taken, there is no conflict; the spacetime is geodesically complete exactly because it evades the BGV hypothesis in the infinite-past limit.

However, this highlights a conceptual tension: at all times $H > 0$ and over any finite interval $H_{avg} > 0$.  In fact, depending on the behavior of the scale factor, how intervals are selected in \eq{eq:havg bgv}, and how the limiting process is executed, one can calculate a \emph{continuum} of values for $H_{avg}$ including both zero and positive values (in this case up to $\nicefrac{2}{\alpha}$). Without further clarification, the BGV theorem may both apply or fail to apply to the same spacetime, depending solely on arbitrary choices of integration domain--despite the geodesic completeness of the spacetime having been definitively established by Thm. \ref{ledthm}. 

 While \cite{Borde:2001nh} does not discuss future completeness, we may consider this case by calculating \eq{eq:havg bgv} with $t_i = 0$, and $t_f>0$:
\begin{equation}
     H_{avg}^{\gamma}=  \frac{1}{t_f} \int_{0}^{t_f} H dt = \frac{1}{t_f} \ln \left( \frac{e^{\nicefrac{2t_f}{\alpha}} + c}{1 + c}\right) \,,\label{eq:havg non lim inf future}
\end{equation}
and taking the limit $t_f \rightarrow +\infty$, yielding $H_{avg} = 2/\alpha$; thus, for $\alpha>0$, $H_{avg}>0$, and yet the spacetime is future complete. Hence, a positive average expansion rate in the future does not necessarily lead to future geodesic incompleteness. This observation highlights that the BGV theorem conclusions are not symmetrical with respect to time direction.  
Further concerns pertaining to the above are detailed in \cite{Lesnefsky:2022fen}. 
For recent developments in this area see \cite{Kinney:2023urn,Geshnizjani:2023hyd,Garcia-Saenz:2024ogr}.

\emph{Implications of LED}--We now turn to several important consequences that follow from Thm. \ref{ledthm}.  We begin with a proposition:
\begin{prop} \label{cor: complete are inflationary}
Every geodesically complete GFRW spacetime with a smooth, non-constant scale factor $a(t)$ must undergo accelerated expansion for at least some period of time.
\end{prop}
This is a purely geometric result, and applies to any metrical theory of gravity. The LED theorem provides necessary and sufficient conditions for null
geodesic completeness in FRW spacetimes
in terms of integrability conditions on the scale factor. For a nonstatic FRW universe, past and future completeness
require that the affine parameter diverge as $t \to \pm\infty$, which
in turn constrains the asymptotic behavior $a(t)$.
In the absence of a phase with $\ddot a>0$, these integrability
conditions cannot be satisfied simultaneously unless the scale factor
is constant. Therefore, any nontrivially evolving geodesically complete
FRW spacetime must contain at least one epoch of accelerated expansion.
A complete proof is given in \cite{Easson:2024fzn}. We emphasize that Proposition~3 applies to nonstatic FRW cosmologies with genuinely evolving scale factor.

Given the direct geodesic equation solution method of Thm. \ref{ledthm}, one arrives at a paradigm shifting pronouncement: non-trivially evolving geodesically complete universes must experience inflation-like behavior. Here we use the term ``inflationary" synonymously with ``accelerated expansion". A detailed exploration of this proposition and its proof, including explicit constructions of geodesically complete bouncing spacetimes, is provided in our companion work \cite{Easson:2024fzn}.

As a direct application of Prop. \ref{cor: complete are inflationary} one reaps the following:
\begin{conj} \label{conj: gr complete nec violation}
In General Relativity, every smooth, non-constant scale factor $a(t)$ of a geodesically complete, flat FRW spacetime must violate the NEC during at least some period of time.
\end{conj}

This hypothesis is not entirely surprising given the well-known fact that the NEC must be violated in order to achieve a cosmological bounce (in a flat $k=0$) FRW spacetime \cite{Hawking:1970zqf}. During a bounce, the universe transitions from a contracting phase ($H<0$) to an expanding phase ($H>0$). This transition inherently requires that at the point where the contraction halts and expansion begins, the derivative of the Hubble parameter, $\dot H$, must be positive.
Since $\dot H = - 4 \pi G (\rho + p)$, we must have $\rho + p<0$, signaling violation of the NEC. Thus, any such spacetime which exhibits a bounce, or bounces, for any part of it's history must violate NEC, and since SEC implies the NEC,  its infringement is inevitable, thereby permitting the possibility of accelerated expansion.~\footnote{With non-zero curvature $k$, one can produce a bounce without violating NEC since, $\dot H = - 4 \pi G (\rho + p) + k/a^2$, can become positive at the bounce due to $k$. Since $\ddot a/a = H^2 + \dot H = - 4 \pi G (\rho/3 + p)$, and $H=0$ at the bounce, $\ddot a/a>0$ and the SEC is violated. Such curvature bounces are discussed in \cite{1978SvAL....4...82S,Graham:2011nb, Burwig:2025hrr}.}

We may further surmise:
\begin{conj} \label{conj: gfrw complete bounce inflate}
Every geodesically complete eternal spacetime which admits a neighborhood isometric to a GFRW with a non-constant scale factor will inflate for some time and bounce at least once, where said bounce may be at infinity.
\end{conj}

Importantly, this perspective reveals that cosmological bounces do not, by themselves, resolve the geodesic incompleteness of inflationary models highlighted by the BGV theorem (Thm.~\ref{thm:bgv}). Rather, geodesic completeness in reasonable nontrivial cosmological spacetimes appears to require both a bounce and a period of inflationary expansion. These are not mutually exclusive phenomena but instead represent complementary features of nonsingular, complete cosmic histories.

\emph{Conclusions}--
Thm. \ref{ledthm}, derived from the direct integration of the geodesic equations, yields a definitive and rigorous criterion for establishing geodesic completeness in FRW (GFRW) spacetimes. Unlike approaches based on averaged quantities such as \( H \), \( H_{\mathrm{avg}} \), or the BGV Thm.~\ref{thm:bgv}, this method furnishes unambiguous and concrete conditions applicable across a wide range of geometries, including nonsingular and eternally inflating models.

Considering the above findings within the context of the incompleteness arguments presented in \cite{Borde:2001nh}, we arrive at a compelling shuffling of logic: The issue is not that inflationary spacetimes are necessarily incomplete; instead, we find that for nonstatic spacetimes to be complete, they must exhibit inflationary behavior. Within classical GR, and with the caveat with respect to traditional energy conditions, accelerated expansion and inflation  play a critical role in resolving the initial singularity problem and, modulo quantum effects, we have shown inflation can be eternal into the past. 

\emph{Acknowledgments}--It is a pleasure to thank R. Brandenberger, P. Davies, G. Ellis, G. Geshnizjani, J. Hua, W. Kinney, B. Kotschwar, A.~Vikman and S. Watson for useful discussions and correspondence. DAE is supported in part by the U.S. Department of Energy, Office of High Energy Physics, under Award Number DE-SC0019470.   JEL is supported by the gracious patronage of Ms. Deborah L Nelson, Dr. Edward J Lesnefsky, and Mrs. Laken S Lesnefsky.

\appendix
\section{Appendix: Controlled SNEC Violation}\label{appsnecplusc}

In this appendix we analyze the closed FRW model with $a(t)=c+a_0 e^{2t/\alpha}$ and show that, while $T_{kk}$ becomes negative at late times, the violation of the smeared null energy condition is uniformly bounded. Moreover, any widening smearing restores positivity, while the ANEC is satisfied in the strongest sense.

\begin{lem}[Pointwise bounds of $T_{kk}$ for the $k=1$ ``$+c$'' model]
Consider \eq{aplusc}, with $c>0$ and $\alpha\neq 0$, and take
affinely parametrized radial null geodesic with $k^\mu$ normalized so that
$k^t=1/a$. Then
\begin{equation}
T_{kk}(t)\equiv T_{\mu\nu}k^\mu k^\nu
=\frac{\rho+p}{a^2}
=\frac{\,2-\dfrac{8cX}{\alpha^2}\,}{(c+X)^4}\, 
\end{equation}
where $X\equiv a_0 e^{2t/\alpha}.$
$T_{kk}$ has a unique global minimum attained at 
$X_\star=(\alpha^2+c^2)/(3c)$, and
\begin{equation}
T_{kk}^{\min}
= T_{kk}(X_\star)
=-\frac{54\,c^{4}}{\alpha^{2}\!\left(\alpha^{2}+4c^{2}\right)^{3}}\,.
\end{equation}
Moreover,
\begin{equation}
\lim_{t\to-\infty}T_{kk}(t)=\frac{2}{c^4}>0,\qquad
\lim_{t\to+\infty}T_{kk}(t)=0^{-}.
\end{equation}
\end{lem}

\begin{proof}
Using $H=\dot a/a= 2X\big/\!\bigl[\alpha(c+X)\bigr]$ and 
$\dot H=4cX\big/\!\bigl[\alpha^2(c+X)^2\bigr]$, one has 
$\rho+p=-2\dot H+2/a^2=\bigl(2-8cX/\alpha^2\bigr)/(c+X)^2$.
With $k^t=1/a$ this yields the stated $T_{kk}=(\rho+p)/a^2$.
Differentiate with respect to $X$:
\[
\frac{dT_{kk}}{dX}
=\frac{8\,\bigl(3cX-\alpha^2-c^2\bigr)}{\alpha^2\,(c+X)^5},
\]
so the unique critical point is $X_\star=(\alpha^2+c^2)/(3c)$. It is a
global minimum (denominator $>0$, numerator changes sign from negative to
positive). Substituting $X_\star$ gives the quoted $T_{kk}^{\min}<0$.
The limits follow from $X\to 0$ as $t\to-\infty$ and $X\to\infty$ as
$t\to+\infty$.
\end{proof}

\begin{prop}[Uniform bound for smeared averages (controlled SNEC violation)]
Let $f\in C^\infty_c(\mathbb{R})$ be nonnegative with $\int f(\lambda)\,d\lambda=1$,
and $\lambda$ the affine parameter along the null geodesic above. Then
\begin{equation}
\int_{-\infty}^{\infty} f(\lambda)\,T_{kk}(\lambda)\,d\lambda
\;\ge\; T_{kk}^{\min}
=-\frac{54\,c^{4}}{\alpha^{2}\!\left(\alpha^{2}+4c^{2}\right)^{3}}\,.
\end{equation}
In particular, any SNEC violation is quantitatively bounded from below by a
finite model-dependent constant.
\end{prop}

\begin{proof}
Since $f\ge 0$ and $\int f=1$,\\
$\int f\,T_{kk}\ge \bigl(\inf T_{kk}\bigr)\int f = T_{kk}^{\min}$.
\end{proof}

\begin{prop}[Positivity for widening compactly supported smearings]
\label{prop:widening-positive}
Fix a center $\lambda_0\in\mathbb{R}$ and a nonnegative profile
$\phi\in C^\infty_c([-1,1])$ with $\int_{-1}^{1}\phi=1$ and
\begin{equation}\label{eq:left-mass}
\int_{-1}^{0}\phi(u)\,du \;>\; 0 \qquad\text{(nonzero past weight).}
\end{equation}
For $\ell>0$ set $f_\ell(\lambda)=\ell^{-1}\phi\!\big((\lambda-\lambda_0)/\ell\big)$.
Then
\begin{equation}\label{eq:dcv-limit}
\lim_{\ell\to\infty}\int_{-\infty}^{\infty} f_\ell(\lambda)\,T_{kk}(\lambda)\,d\lambda
\;=\; \frac{2}{c^4}\int_{-1}^{0}\phi(u)\,du \;\ge\;0\,.
\end{equation}
In particular, there exists $\ell_\star=\ell_\star(\lambda_0,\phi,c,\alpha)$ such that
\begin{equation}
\int f_\ell(\lambda)\,T_{kk}(\lambda)\,d\lambda\;\ge\;0
\qquad\text{for all }\ \ell\ge \ell_\star.
\end{equation}
\end{prop}

\begin{proof}
Change variables $\lambda=\lambda_0+u\ell$. Since $f_\ell(\lambda)\,d\lambda=\phi(u)\,du$, we have
\begin{equation}
\int f_\ell T_{kk} \;=\; \int_{-1}^{1} \phi(u)\,T_{kk}\!\big(\lambda_0+u\ell\big)\,du.
\end{equation}
By Lemma~1, $T_{kk}(\lambda)\to 2/c^4$ as $\lambda\to-\infty$ and $T_{kk}(\lambda)\to 0^-$ as $\lambda\to+\infty$, while $T_{kk}$ is bounded below by $T_{kk}^{\min}$ and above by $2/c^4$. Hence for each fixed $u<0$, $\lambda_0+u\ell\to-\infty$ and $T_{kk}\!\big(\lambda_0+u\ell\big)\to 2/c^4$; for $u>0$, $\lambda_0+u\ell\to+\infty$ and $T_{kk}\!\big(\lambda_0+u\ell\big)\to 0^-$. The integrand is dominated by an $L^1$ function independent of $\ell$, so dominated convergence applies and yields \eqref{eq:dcv-limit}. The tail positivity then follows for all sufficiently large $\ell$.
\end{proof}

\begin{rmk}
Because $T_{kk}(t)<0$ for all sufficiently late times (indeed for $t>t_0$, where
$X=a_0 e^{2t/\alpha}>\alpha^2/(4c)$, i.e.\ $t_0=\frac{\alpha}{2}\ln\!\frac{\alpha^2}{4ca_0}$), the SNEC is violated: one can choose $f$ supported entirely where $T_{kk}<0$. The results
above show that this violation is controlled: $T_{kk}$ is bounded
below by the explicit constant $T_{kk}^{\min}$, and any standard widening
family of compactly supported smearings becomes nonnegative once the window
includes a sufficient portion of the positive past tail. Meanwhile, the ANEC
holds in the strongest sense since $T_{kk}\to 2/c^4$ as $\lambda\to-\infty$
and $d\lambda\sim a\,dt\sim c\,dt$ there, so 
$\int T_{kk}\,d\lambda=+\infty$.
\end{rmk}

\end{document}